\documentclass[12pt,a4paper]{article}
\usepackage{a4wide}
\usepackage{showlabels}
\usepackage[utf8]{inputenc}
\usepackage{enumerate}
\usepackage{amsmath}
\usepackage{amsfonts}
\usepackage{nicefrac}
\usepackage{amssymb}
\usepackage{amsthm}
\usepackage{mathtools}
\usepackage{tikz}
\usepackage{xcolor}
\usepackage{authblk}
\usepackage{tabularx}
\usepackage{paralist}
\usepackage{booktabs}
\usepackage{paralist}
\usepackage{hyperref}
\usepackage[square,numbers]{natbib}
\usepackage[capitalize]{cleveref}
\usepackage[backgroundcolor=gray!10,textsize=footnotesize]{todonotes}
\usepackage{etoolbox}
\usepackage{xspace}
\usepackage{algpseudocode}
\usepackage{algorithm}

\usetikzlibrary{decorations.pathreplacing}
\usetikzlibrary{math}
\usetikzlibrary{calc}
\usetikzlibrary{shapes}

\newtheorem{theorem}{Theorem}
\newtheorem{lemma}[theorem]{Lemma}
\newtheorem{observation}[theorem]{Observation}
\newtheorem{corollary}[theorem]{Corollary}

\theoremstyle{definition}

\newcommand{\problemdef}[3]{
	\begin{center}
	\begin{minipage}{0.95\columnwidth}
		\noindent
		\textsc{#1}
		\vspace{5pt}\\
		\setlength{\tabcolsep}{3pt}
		\begin{tabularx}{\textwidth}{@{}lX@{}}
			\textbf{Input:}     & #2 \\
			\textbf{Question:}  & #3
		\end{tabularx}
	\end{minipage}
	\end{center}
}

\newcommand{\MTsum}{\textsc{MinTimeline}$_+$\xspace}
\newcommand{\MTmax}{\textsc{MinTimeline}$_\infty$\xspace}

\newcommand{\XP}{\textrm{XP}\xspace}
\newcommand{\Wone}{\textrm{W[1]}\xspace}
\newcommand{\FPT}{\textrm{FPT}\xspace}
\newcommand{\NP}{\textrm{NP}\xspace}

\newcommand{\G}{\mathcal{G}}
\newcommand{\TGcompact}{\G := (V,(E_i)_{i=1}^\tau)}
\newcommand{\vc}{\mathcal{T}}

\newcommand{\vcs}{\mathcal{C}}

\DeclareMathOperator{\inc}{inc}

\Crefname{theorem}{Theorem}{Theorems}
\crefname{theorem}{Thm.}{Thms.}
\Crefname{corollary}{Corollary}{Corollaries}
\crefname{corollary}{Cor.}{Cors.}
\newcommand{\tcite}[1]{\textsuperscript{\tiny[\cref{#1}]}}

\title{Disentangling the Computational Complexity of Network Untangling\thanks{This work was initiated at the research retreat of the Algorithmics and Computational Complexity group, TU Berlin, held at
Zinnowitz, September 2021.}}

\author{Vincent Froese}
\author{Pascal Kunz\thanks{Supported by the DFG Research Training Group 2434 ``Facets of Complexity''.}}
\author{Philipp Zschoche}

\affil{\small
  Technische Universit\"at Berlin, Faculty~IV, Institute of Software Engineering and Theoretical Computer Science, Algorithmics and Computational Complexity.\protect\\
  \texttt{\{vincent.froese,p.kunz.1,zschoche\}@tu-berlin.de}}

\date{}

\begin{document}

\maketitle

\begin{abstract}
  We study the network untangling problem introduced by Rozenshtein, Tatti, and Gionis~[DMKD 2021], which is a variant of~\textsc{Vertex Cover} on temporal graphs---graphs whose edge set changes over discrete time steps. They introduce two problem variants. The goal is to select at most~$k$ time intervals for each vertex such that all time-edges are covered and (depending on the problem variant) either the maximum interval length or the total sum of interval lengths is minimized. This problem has data mining applications in finding activity timelines that explain the interactions of entities in complex networks.

  Both variants of the problem are \NP-hard. In this paper, we initiate a multivariate complexity analysis involving the following parameters: number of vertices, lifetime of the temporal graph, number of intervals per vertex, and the interval length bound. For both problem versions, we (almost) completely settle the parameterized complexity for all combinations of those four parameters, thereby delineating the border of fixed-parameter tractability.
\end{abstract}

\section{Introduction}
The classical \textsc{Vertex Cover} problem is among the most well-studied \NP-hard problems on static graphs especially in the context of parameterized algorithmics. 
In fact, it is often referred to as the \emph{Drosophila} of parameterized complexity 
\cite{DF13,%
FellowsJKRW18}.
Real-world graphs, however, are often dynamic and change over time.
In modern applications, this temporal information is readily available and allows for more realistic models.
This leads algorithmic research to focus on temporal graphs 
\cite{BoehmerFHLNR21,%
iwocaBumpusM21,%
DeligkasP20,%
jcssErlebach0K21,%
FanJHYWWMX21,%
HS12,%
tcsMichailS16,%
SingerGR19}.
In a temporal graph, edges can appear and disappear over time, which is formalized by a sequence~$E_1,\ldots,E_\tau$ of edge sets over a fixed set~$V$ of vertices, where $\tau\ge 1$ denotes the lifetime of the temporal graph.
Several classical graph problems have been studied on temporal graphs \cite{AMSZ20,algorithmicaCasteigtsHMZ21,ijcaiKlobasMMNZ21,%
stacsMertziosMNZZ20,jcssMertziosMZ21,%
jcssZschocheFMN20}.
  
In this paper, we study the network untangling problem introduced by Rozenshtein et al.~\cite{RTG17,RTG21}, which is a temporal variant of \textsc{Vertex Cover} and motivated by data mining applications such as discovering event timelines and summarizing temporal networks.
Here, edges in the temporal graph model certain interactions between entities (vertices).
The goal is to explain the observed interactions by selecting few (and short) activity intervals for each entity such that at the time of an interaction at least one of the two entities is active.
The formal definition is as follows:
Let $\G=(V,(E_i)_{i\in[\tau]})$ be a temporal graph.
A~\emph{$k$-activity timeline} is a set~$\vc$ containing at most~$k$ time intervals for each vertex, that is, $\vc\subseteq \{(v,a,b)\in V\times [\tau]\times[\tau]\mid a\le b\}$ such that $|\{(v,a,b)\in\vc\}|\le k$ for each~$v\in V$.
We say that~$\vc$ \emph{covers}~$\G$ if for each~$t\in[\tau]$ and each~$\{u,v\}\in E_t$, $\vc$ contains some~$(u,a,b)$ with~$t\in[a,b]$ or some $(v,a,b)$ with~$t\in[a,b]$.
That is, for each time step~$t$, the set~$\{v\in V\mid (v,a,b)\in\vc, t\in [a,b]\}$ must be a vertex cover for the graph~$G_t=(V,E_t)$.
The task is to find a~$k$-activity timeline that minimizes some objective regarding the interval lengths.
\citet{RTG21} introduced the following two problems.

\problemdef{\MTmax}
{A temporal graph $\G=(V,E_1,\ldots,E_\tau)$ and $k,\ell\in\mathbb{N}_0.$}
{Is there a~$k$-activity timeline~$\vc$ covering~$\G$ with $\max_{(v,a,b)\in\vc}(b-a)\le \ell$?}

\problemdef{\MTsum}
{A temporal graph $\G=(V,E_1,\ldots,E_\tau)$ and $k,\ell\in\mathbb{N}_0.$}
{Is there a~$k$-activity timeline~$\vc$ covering~$\G$ with $\sum_{(v,a,b)\in\vc}(b-a)\le \ell$?}
\noindent
An example instance and solutions for each of these two problems are pictured in \Cref{fig:intro-ex}.
Note that both problems are equivalent if~$\ell=0$.

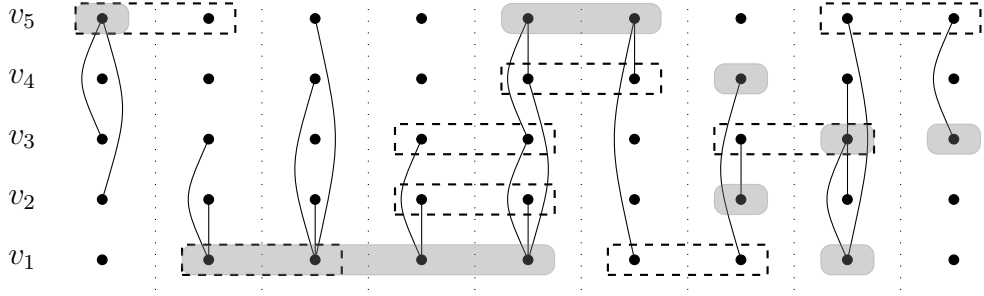
\begin{figure}[t]
	\centering
	\begin{tikzpicture}[xscale=1.4,yscale=0.8]
	\def\nsc{0.33}
	\tikzstyle{xnode}=[circle,scale=\nsc,draw,fill=black];
	\tikzstyle{max}=[dashed, thick];
	\tikzstyle{sum}=[rounded corners,gray,opacity=0.35,fill];

	\node () at (0.25,1) {$v_1$};
	\node () at (0.25,2) {$v_2$};
	\node () at (0.25,3) {$v_3$};
	\node () at (0.25,4) {$v_4$};
	\node () at (0.25,5) {$v_5$};
	
	\foreach \x in {1,...,8} {
		\draw[loosely dotted] (\x+0.5,0.5) -- (\x+0.5,5.5);
	}
	
	\foreach \x in {1,...,9} {
		\foreach \y in {1,...,5} {
			\node (\x!v\y) at (\x,\y)[xnode]{};
		}
	}
	\draw (1!v2) .. controls (1.25,3.5) .. (1!v5);
	\draw (1!v3) .. controls (0.75,4) .. (1!v5);
	\draw (2!v1) -- (2!v2);
	\draw (2!v1) .. controls (1.75,2) .. (2!v3);
	\draw (3!v1) -- (3!v2);
	\draw (3!v1) .. controls (2.75,2.5) .. (3!v4);
	\draw (3!v1) .. controls (3.25,3) .. (3!v5);
	\draw (4!v1) -- (4!v2);
	\draw (4!v1) .. controls (3.75,2) .. (4!v3);
	\draw (5!v1) -- (5!v2);
	\draw (5!v1) .. controls (4.75,2) .. (5!v3);
	\draw (5!v1) .. controls (5.25,2.5) .. (5!v4);
	\draw (5!v4) -- (5!v5);
	\draw (5!v3) .. controls (4.75,4) .. (5!v5);
	\draw (6!v1) .. controls (5.75,3) .. (6!v5);
	\draw (6!v4) -- (6!v5);
	\draw (7!v2) -- (7!v3);
	\draw (7!v1) .. controls (6.75,2.5) .. (7!v4);
	\draw (8!v1) .. controls (7.75,2) .. (8!v3);
	\draw (8!v1) .. controls (8.25,3) .. (8!v5);
	\draw (8!v2) -- (8!v3);
	\draw (8!v3) -- (8!v4);
	\draw (9!v3) .. controls (8.75,4) .. (9!v5);

	\draw[max] (0.75, 4.75) rectangle (2.25, 5.25) {};
	\draw[max] (1.75, 0.75) rectangle (3.25, 1.25) {};
	\draw[max] (3.75, 1.75) rectangle (5.25, 2.25) {};
	\draw[max] (3.75, 2.75) rectangle (5.25, 3.25) {};
	\draw[max] (4.75, 3.75) rectangle (6.25, 4.25) {};
	\draw[max] (5.75, 0.75) rectangle (7.25, 1.25) {};
	\draw[max] (6.75, 2.75) rectangle (8.25, 3.25) {};
	\draw[max] (7.75, 4.75) rectangle (9.25, 5.25) {};

	\draw[sum] (0.75, 4.75) rectangle (1.25, 5.25) {};
	\draw[sum] (1.75, 0.75) rectangle (5.25, 1.25) {};
	\draw[sum] (4.75, 4.75) rectangle (6.25, 5.25) {};
	\draw[sum] (6.75, 1.75) rectangle (7.25, 2.25) {};
	\draw[sum] (6.75, 3.75) rectangle (7.25, 4.25) {};
	\draw[sum] (7.75, 0.75) rectangle (8.25, 1.25) {};
	\draw[sum] (7.75, 2.75) rectangle (8.25, 3.25) {};
	\draw[sum] (8.75, 2.75) rectangle (9.25, 3.25) {};
	\end{tikzpicture}
	\caption{
		An example of a temporal graph:
		The areas between the dotted lines represent the individual layers.
		The intervals highlighted in gray form a solution for \MTsum (with $\ell=4$, $k=2$).
		The intervals indicated by dashed boxes form a solution for \MTmax (with $\ell=1$, $k=2$).
	}
	\label{fig:intro-ex}
\end{figure}

\citet{RTG21} showed that both problems are \NP-hard and hard to approximate in polynomial time.
They focus on developing heuristics and show in a case study that activity timelines can successfully be recovered for real-world data of Twitter users.
In this work, we provide a deeper complexity-theoretic understanding of both problems
through the lens of parameterized complexity theory.
To this end, we study the influence of several natural problem parameters on the algorithmic complexity, namely, the number~$n$ of vertices, the lifetime~$\tau$, the maximum number~$k$ of intervals per vertex, and the interval length bound~$\ell$.
In our multivariate analysis, we identify for both problems which parameter combinations yield fixed-parameter tractability and which combinations are intractable.
In doing so, we reveal interesting connections to seemingly unrelated problems like graph coloring and bin packing.
Our results yield an (almost) tight characterization and pave the way for a comprehensive picture of the computational complexity landscape of network untangling.

\paragraph{Related Work.}
As the literature on edge covering is extremely rich,
we only consider studies closely related to our setting. 
For a broader overview on the topic of temporal network mining, we refer to a recent tutorial~\cite{RG19}.
Our main reference is the work by~\citet{RTG21} who introduced both problems and showed that \MTmax is polynomial-time solvable for~$k=1$, whereas \MTsum is \NP-hard for~$k=1$.
Moreover, they showed that both problems are \NP-hard for~$\ell=0$ with unbounded~$k$ and thus not approximable in polynomial time.
They develop efficient heuristics to solve the problems and provide experiments to evaluate the performance of their approaches.

\citet{AMSZ20} studied a different variant of \textsc{Vertex Cover} on temporal graphs.
Their model expects an edge to be covered at least once over every time window of some given size~$\Delta$.
That is, 
they define a temporal vertex cover as a set~$S\subseteq V\times [\tau]$ such that,
for every time window of size~$\Delta$ and for each edge~$e=\{v,w\}$ appearing in a time step contained in the time window,
it holds that~$(v,t)\in S$ or~$(w,t)\in S$ for some~$t$ in the time window with time-edge~$(e,t)$.
Among other results, they provide \NP-hardness results in very restricted settings. 

\citet{FNRZ19} and \citet{tcsHeegerHKNRS21} 
studied the parameterized complexity of a ``multistage'' variant of \textsc{Vertex Cover}.
Here, one is given a temporal graph
and seeks a sequence
of vertex covers of size at most~$k$ 
(one for each layer of the temporal graph) 
such that consecutive vertex covers are in some sense similar,
e.g., the symmetric difference of two consecutive vertex covers is upper-bounded by some value~\cite{FNRZ19} or
the sum of symmetric differences of all consecutive vertex covers is upper-bounded by some value~\cite{tcsHeegerHKNRS21}.

\paragraph{Our Contributions.}

We almost completely settle the parameterized complexity of the two variants \MTmax and \MTsum for all combinations of the parameters~$n$, $\tau$, $k$, and~$\ell$.
\Cref{tab:results} gives an overview of the central results.
It turns out that both problems are already \NP-hard if the temporal graph contains three identical layers and $k=2$ and~$\ell=0$~(\Cref{cor:nphard-tau-k-l}). This strengthens the hardness result by~\citet{RTG21} for~$\ell=0$
and implies that fixed-parameter tractability can only be achieved in combination with the parameter~$n$.
Interestingly, the two problems behave somewhat differently here.
Our two main results state that \MTmax parameterized by~$n$ is \Wone-hard even for~$\ell=1$~(\Cref{thm:w1-max-n}), while \MTsum parameterized by~$n+\ell$ is in \FPT~(\Cref{thm:sum-fpt-n+l}).
We further show that both problems are in \XP when parameterized by~$n$ (\Cref{thm:XP-n-max,thm:FPT-n+k-sum}) and in \FPT if~$\ell=0$~(\Cref{thm:ell=0FPTn}).
\Cref{thm:FPT-n+k-sum} shows that \MTsum is also in \FPT for~$n+k$ and
\Cref{thm:fpt-n+k-max} shows the same for \MTmax.
As regards the case~$k=1$, we strengthen the \NP-hardness of \MTsum showing that it already holds for two layers~(\Cref{thm:tau=2}).
Note that~$n+\tau$ trivially yields fixed-parameter tractability
since the instance size is bounded in these numbers.
Thus, for \MTmax, we obtain a full complexity dichotomy.
For \MTsum, the parameterized complexity regarding~$n$ remains open.
Notably, our hardness results also imply running time lower bounds based on the Exponential Time Hypothesis\footnote{The Exponential Time Hypothesis~\cite{IP01} states that \textsc{3-SAT}
cannot be solved in~$O(2^{cn})$ time for some constant~$c>0$, where~$n$ is the number of variables of the input formula.} (ETH)~(\Cref{cor:ETHcoloring,cor:ETHbin}).

\begin{table*}[t]
	\centering
	\caption{Overview of results. Parameters: $n$ number of vertices, $\tau$ lifetime, $k$ number of intervals per vertex, $\ell$ interval length bound.}
	\begin{tabular}{lrr}
          \toprule
          Parameter & \MTmax & \MTsum \\
          \midrule
          $n+k$ & \FPT{}\tcite{thm:fpt-n+k-max} & \FPT{}\tcite{thm:FPT-n+k-sum}\\
          $n+\ell$ & \Wone-h.\tcite{thm:w1-max-n}, \XP\tcite{thm:XP-n-max} & \FPT{}\tcite{thm:sum-fpt-n+l}\\
          $n$ & \Wone-h.\tcite{thm:w1-max-n} ($\ell=1$), \XP{}\tcite{thm:XP-n-max} & \XP{}\tcite{thm:FPT-n+k-sum} \\
          $\tau=3$, $k=2$, $\ell=0$ & \NP-h.\tcite{cor:nphard-tau-k-l} & \NP-h.\tcite{cor:nphard-tau-k-l}\\
          \bottomrule
	\end{tabular}%
	\label{tab:results}
\end{table*}

\section{Preliminaries}
We denote the positive integers by~$\mathbb N$ and let $\mathbb{N}_0\coloneqq \mathbb{N} \cup \{0\}$.
For~$n\in\mathbb{N}$, let~$[n]\coloneqq \{1,\ldots,n\}$.

\paragraph*{Graphs.}
If not specified otherwise, a graph is static and undirected.
We refer to \citet{Die16} for details on notation in graph theory.
If $G=(V,E)$ is a graph and $S\subseteq V$, let $\inc(S)\coloneqq \{ e \in E \mid e \cap S \neq \emptyset\}$ denote the set of edges incident to at least one vertex in $S$.
A \emph{temporal graph} $\TGcompact$ consists of a finite vertex set $V$ and a sequence of edge sets $E_1,\ldots,E_\tau \subseteq \binom{V}{2}$.
The pair~$(e,i)$ is a \emph{time-edge} of $\G$ if $e \in E_i$.
The graph~$(V,E_i)$ is called the~$i$-th layer of $\G$.
The \emph{size} of $\G$ is $|\G| \coloneqq |V| + \sum_{t=1}^\tau \max\{1,|E_t|\}$. 

\paragraph*{Parameterized Complexity.} 
Let~$\Sigma$ denote a
finite alphabet.
A \emph{parameterized problem}~$L\subseteq \{(x,k)\in \Sigma^*\times \mathbb N_0\}$ is a subset of all instances~$(x,k)$ in~$\Sigma^*\times \mathbb N_0$,
where~$k$ denotes the \emph{parameter}.
A parameterized problem~$L$ is 
\begin{inparaenum}[(i)]
\item \emph{fixed-parameter tractable} (or contained in the class \FPT) if there is an algorithm that decides every instance~$(x,k)$ for~$L$ in~$f(k)\cdot |x|^{O(1)}$ time,  
 \item contained in the class \XP if there is an algorithm that decides every instance~$(x,k)$ for~$L$ in~$|x|^{f(k)}$ time, and
 \item \emph{para-NP-hard} if $L$ is~\NP-hard for some constant value of the parameter,
\end{inparaenum}
where~$f$ is any computable function that only depends on the parameter.
Note that $\FPT\subseteq \XP$.
If a parameterized problem is \emph{\Wone-hard}, 
then it is presumably not in \FPT, and if it is para-NP-hard, then it is not in~\XP (unless P$\;=\;$\NP).
Further details can be found in~\cite{DF13}.

\section{First Results}
\label{sec:first-results}

This section sets the basis of our parameterized complexity analysis containing results
for most of the parameter combinations (except for~$n+\ell$).

We start with the general observation (which we will use later) that, for~$\ell=0$, the temporal order of the layers is irrelevant since every element in a $k$-activity timeline
covers only edges of one layer.
\begin{observation}
	\label{obs:permutation}
	Let $\pi\colon [\tau] \rightarrow [\tau]$ be a permutation.
	Then, the instance $((V,(E_i)_{i\in[\tau]}),k,0)$ is a yes-instance of \MTmax (\MTsum) 
	if and only if $((V,(E_{\pi(i)})_{i\in[\tau]}),k,0)$ is a yes-instance.
\end{observation}

\noindent
Moreover, if all layers of the temporal graph are equal, 
then, for~$\ell=0$, both problems are equivalent to an extension of $k$-coloring.
For a graph $G=(V,E)$ and $a,b \in \mathbb{N}$ with $a \geq b$,
an \emph{$(a{:}b)$-coloring} of~$G$ is a function $c \colon V \rightarrow \binom{[a]}{b}$ such that $c(u) \cap c(v) = \emptyset$ for all $\{u,v\} \in E$.
Note that a $(k{:}1)$-coloring of a graph is simply a $k$-coloring.
\newcommand{\abColoring}[1]{\textsc{$(#1)$-Coloring}}%
Given a graph $G$, the problem \abColoring{a{:}b} asks 
whether there is an $(a{:}b)$-coloring for~$G$.
We get the following.
\begin{theorem}
  \label{thm:coloring-equivalence}
  If $E_1=E_2=\dots=E_\tau$, ~$\ell=0$, and $\tau \geq k$, then \MTmax (\MTsum) is equivalent to~\textsc{$(\tau{:}\tau-k)$-Coloring}.  
\end{theorem}
\begin{proof}
	Suppose that $\vc$ is a $k$-activity timeline that covers $\G=(V,E_1,\ldots,E_\tau)$ with $E_1=\dots=E_\tau \eqqcolon E$ and $\ell=0$.
	We can assume that $\vc$ contains exactly~$k$ intervals for each~$v\in V$.
	Let $G=(V,E)$.
	Then, $c \colon V \rightarrow \binom{[\tau]}{\tau-k}$ with $c(v) \coloneqq \{t \in [\tau] \mid (v,t,t) \notin \vc \}$ is a $(\tau{:}\tau-k)$-coloring of $G$.
	
	Conversely, if $c \colon V \rightarrow \binom{[\tau]}{\tau-k}$ is a $(\tau{:}\tau -k)$-coloring of $G=(V,E)$, then~$\vc \coloneqq \{ (v,t,t)\in V\times[\tau]\times[\tau] \mid v\in V, t \notin c(v) \}$ is a $k$-activity timeline that covers $\G$.
\end{proof}

\noindent
The \NP-hardness of \textsc{3-Coloring}, or \textsc{$(3{:}1)$-Coloring}, implies the following.

\begin{corollary}
  \label{cor:nphard-tau-k-l}
  \MTmax and \MTsum are \NP-hard, even if $k=2$, $\ell=0$, $\tau=3$, and all layers are identical.
\end{corollary}

\noindent
\citet{BKPSW19} mention a result by~\citet{Nederlof08} showing that \textsc{$(a{:}b)$-Coloring} can be solved in~$(b+1)^n\cdot n^{O(1)}$ time. They show that the ETH implies that there is no~$2^{o(\log b)n}$-time algorithm for \abColoring{a{:}b}~\cite[Theorem~1.1]{BKPSW19}.
Thus, we have the following.

\begin{corollary}\label{cor:ETHcoloring}
  \MTmax (\MTsum) with $E_1=E_2=\dots=E_\tau$ and~$\ell=0$ cannot be solved in~$f(\tau-k)\cdot 2^{o(\log(\tau-k))n}$ time for any function~$f$ unless the ETH fails.
\end{corollary}

\noindent
Note that, if $\tau = 2$, then \MTmax and \MTsum are both trivial unless $k=1$.
\citet{RTG21} already showed that \MTmax is polynomial-time solvable for $k=1$. 
By contrast, we show that \MTsum is \NP-hard in the case of~$k=1$ and~$\tau=2$ (strengthening the \NP-hardness for~$k=1$ with unbounded~$\tau$ by \citet[Proposition~3]{RTG21}).
Moreover, \MTmax is \FPT when parameterized by~$\ell$ with this restriction.
This gives us the following.
\begin{theorem}
	\label{thm:tau=2}
	\MTsum with two layers is \NP-hard and \FPT for parameter~$\ell$.
\end{theorem}
\begin{proof}
  We show \NP-hardness with a reduction from \textsc{Odd Cycle Transversal}.
	In this problem, the input consists of a graph $G=(V,E)$ and an integer~$s$ and the task is to decide whether there is a vertex set~$X\subseteq V$ with~$|X| \leq s$ such that~$G-X$ is bipartite.
	This problem is \NP-hard~\cite{LY80}.
	
	Given an instance $(G=(V,E),s)$, the reduction outputs the instance $(\G\coloneqq(V,E_1 \coloneqq E,E_2 \coloneqq E),k \coloneqq 1,\ell \coloneqq s)$.
	If $X\subseteq V$ with $|X|\leq s$ is an odd cycle transversal in $G$ such that $V_1$ and $V_2$ are the two color classes in the bipartite graph $G-X$, then $\vc\coloneqq \{(v,1,2) \mid v \in X\} \cup \{ (v,i,i) \mid i \in \{1,2\}, v\in V_i\}$ is a $1$-activity timeline that covers $\G$ with $\sum_{(v,a,b)\in\vc}(b-a)\le s$.
	Conversely, if $\vc$ is a $1$-activity timeline that covers $\G$ with $\sum_{(v,a,b)\in\vc}(b-a)\le s$, then $X \coloneqq \{v \in V \mid (v,1,2) \in \vc\}$ is an odd cycle transversal of size at most $s$ in 
	$G$.

        For fixed-parameter tractability, we show that \MTsum with $\tau = 2$ can be reduced to \textsc{Almost 2-SAT}.
	In this problem, the input consists of a Boolean formula~$\varphi$ in $2$-CNF (with duplicate clauses allowed) and an integer $s$ and the task is to decide whether $\varphi$ can be made satisfiable by deleting at most $s$ clauses.
	The problem is known to be in \FPT when parameterized by~$s$~\cite{RO09}.
	
	Let $(\G=(V,E_1,E_2),k,\ell)$ be an instance for \MTsum.
	If $k \neq 1$, the instance is trivial.
	Otherwise, we introduce two Boolean variables~$x^v_1$ and~$x^v_2$ for each~$v \in V$.
	Intuitively, $x^v_i$ is true if $v$ is used to cover edges in~$E_i$.
	The reduction outputs the instance $(\varphi, s\coloneqq \ell)$, where $\varphi$ is constructed as follows:
	For every edge $\{u,v\} \in E_i$, $i\in\{1,2\}$, $\varphi$ contains $s + 1$ copies of the clause $(x^u_i \vee x^v_i)$ expressing that the edge must be covered.
	For every vertex~$v \in V$, the formula contains the clause~$(\neg x^v_1 \vee \neg x^v_2)$, stating that $v$ can only be used once unless this clause is deleted.

        Now, suppose that $X$ is a set of at most~$\ell$ clauses of $\varphi$ and that $\alpha$ is a satisfying assignment for all other clauses.
	Then, $\vc \coloneqq \{(v,1,2) \mid (\neg x^v_1\vee \neg x^v_2) \in X\} \cup \{(v,i,i) \mid \alpha(x^v_i) = \texttt{true}, (\neg x^v_1\vee \neg  x^v_2) \notin X\}$ is a $1$-activity timeline that covers $\G$ with $\sum_{(v,a,b)\in\vc}(b-a)\le \ell$.
	Conversely, if $\vc$ is a $1$-activity timeline that covers $\G$ with $\sum_{(v,a,b)\in\vc}(b-a)\le \ell$, then removing the at most~$\ell$ clauses in $\{(\neg x^v_1 \vee \neg x^v_2)\mid (v,1,2)\in\vc\}$ from $\varphi$ makes the formula satisfiable by the following assignment: $\alpha(x^v_i) = \texttt{true}$ if and only if $(v,i,i)\in\vc$ or $(v,1,2)\in\vc$.
\end{proof}

Since \MTmax and \MTsum are \NP-hard, even if $\tau+k+\ell$ is constant,
the interesting parameterizations remaining involve the number~$n$ of vertices. 
Clearly, combining~$n$ and~$\tau$ trivially yields fixed-parameter tractability for both problems,
as the overall instance size is bounded in these two parameters.
If we use only the number $n$ of vertices as a parameter, 
then we can use dynamic programming to show 
that both problems are in \XP.

\begin{theorem}
  \label{thm:XP-n-max}
  \MTmax is solvable in $(k+1)^n(\ell+2)^n2^n \tau  n^{O(1)}$ time.
\end{theorem}
\begin{proof}
  We solve an instance~$(\G=(V=\{v_1,\ldots,v_n\},(E_i)_{i\in[\tau]}),k,\ell)$ of \MTmax via dynamic programming.
  We define a Boolean table~$T$ of size~$\tau(k+1)^n(\ell+2)^n$ as follows:
  For each~$i\in[\tau]$, $k_1,\ldots,k_n\in\{0,\ldots,k\}$, $\ell_1,\ldots,\ell_n\in\{-1,0,\ldots,\ell\}$,
  $T[i,k_1,\ldots,k_n,\ell_1,\ldots,\ell_n]= \texttt{true}$ if and only if there exists a~$\vc \subseteq V\times[i]\times[i]$ with~$\max_{(v,a,b)\in\vc}(b-a)\le\ell$ that covers~$(V,E_1,\ldots,E_i)$ and, for each~$j\in[n]$, satisfies
  \begin{compactitem}
    \item $|\{(v_j,a,b)\in\vc\}|\le k_j$,
    \item$|\{(v_j,a,i)\in\vc\}|=0$ if~$\ell_j=-1$, and
    \item $(v_j,i-\ell_j,i)\in\vc$ if~$\ell_j \ge 0$.
  \end{compactitem}
  Note that we have a yes-instance if and only if~$T[\tau,k,\ldots,k,\ell_1,\ldots,\ell_n]=\texttt{true}$ for some~$\ell_1,\ldots,\ell_n$.
  
  For initialization, we set
  $T[1,k_1,\ldots,k_n,\ell_1,\ldots,\ell_n] \coloneqq \texttt{true}$ if and only if
  \begin{compactitem}
    \item $\ell_j\le 0$ for all~$j\in[n]$,
    \item $\ell_j = -1$ for all~$j\in[n]$ with~$k_j=0$,
    \item $k_j > 0$ for all~$\ell_j=0$, and
    \item $\{v_j\mid \ell_j=0\}$ is a vertex cover for~$(V,E_1)$.
  \end{compactitem}
  This is clear from the definition of~$T$.
  The table~$T$ can then be filled recursively as follows:
  
  $T[i,k_1,\ldots,k_n,\ell_1,\ldots,\ell_n] \coloneqq \texttt{true}$ if and only if
  \begin{compactenum}[(i)]
    \item $\ell_j\le i-1$ for all~$j\in[n]$,
    \item $\ell_j = -1$ for all~$j\in[n]$ with~$k_j=0$,
    \item $k_j > 0$ for all~$\ell_j\ge 0$,
    \item $\{v_j\mid \ell_j\ge 0\}$ is a vertex cover for~$(V,E_i)$, and
    \item $T[i-1,k_1',\ldots,k_n',\ell_1',\ldots,\ell_n']=\texttt{true}$ for some~$k_1',\ldots,k_n',\ell_1',\ldots,\ell_n'$, where
      \begin{compactitem}
        \item[$\circ$] $k_j'=k_j$ and~$\ell_j'=\ell_j-1$ if~$\ell_j>0$,
        \item[$\circ$] $k_j'=k_j-1$ and~$\ell_j'\in\{-1,\min(\ell,i-2)\}$ if~$\ell_j=0$, and
        \item[$\circ$] $k_j'=k_j$ and $\ell_j'\in\{-1,\min(\ell,i-2)\}$ if~$\ell_j=-1$.
      \end{compactitem}
    \end{compactenum}

    Conditions~(i)--(iv) ensure that the values~$k_j$ and~$\ell_j$ are consistent and yield a local vertex cover for layer~$i$.
    Condition~(v) checks whether the local solution extends to a valid solution for all previous layers.
    
    Filling an entry of the table can be done with at most~$2^n$ table look-ups. Thus, the overall running time is~$\tau(k+1)^n(\ell+2)^n2^nn^{O(1)}$.
\end{proof}
For \MTsum, one can achieve a better running time with similar dynamic programming---this even yields fixed-parameter tractability for~$n+k$ and \XP for~$n$.

\begin{theorem}
	\label{thm:FPT-n+k-sum}
	\MTsum is solvable in $(k+1)^n2^{O(n)}\tau(\ell+1)$ time.
\end{theorem}
\begin{proof}
 Given an instance~$(\G=(V=\{v_1,\ldots,v_n\},(E_i)_{i\in[\tau]}),k,\ell)$, we define a Boolean table~$T$ of size~$\tau (\ell+1) (k+1)^n2^n$ as follows:
 For each~$i\in[\tau]$, $k_1,\ldots,k_n\in\{0,\ldots,k\}$, $S \subseteq V$, and~$\l\in\{0,\ldots,\ell\}$,
  $T[i,k_1,\ldots,k_n,S,l]= \texttt{true}$ if and only if there exists a~$\vc \subseteq V\times[i]\times[i]$ with~$\sum_{(v,a,b)\in\vc}(b-a)\le l$ that covers~$(V,E_1,\ldots,E_i)$ and, for each~$j\in[n]$, satisfies
  \begin{compactitem}
    \item $|\{(v_j,a,b)\in\vc\}|\le k_j$,
    \item$|\{(v_j,a,i)\in\vc\}|=0$ if~$v_j\not\in S$, and
    \item $|\{(v_j,a,i)\in\vc\}|>0$ if~$v_j \in S$.
    \end{compactitem}
    That is, we have a yes-instance if and only if~$T[\tau,k,\ldots,k,S,\ell]=\texttt{true}$ for some~$S\subseteq V$.
    By definition, we set~$T[1,k_1,\ldots,k_n,S,l]=\texttt{true}$ if and only if
    \begin{compactitem}
      \item $v_j\not\in S$ for all~$j\in[n]$ with~$k_j=0$ and 
      \item $S$ is a vertex cover for~$(V,E_1)$.
    \end{compactitem}
    The remaining table entries can then be computed recursively:

    $T[i,k_1,\ldots,k_n,S,l]=\texttt{true}$ if and only if
    \begin{compactenum}[(i)]
      \item $v_j\not\in S$ for all~$j\in[n]$ with~$k_j=0$,
      \item $S$ is a vertex cover for~$(V,E_i)$, and
      \item $T[i-1,k_1',\ldots,k_n',S',l']=\texttt{true}$ for some $S'\subseteq V$, where
        \begin{compactenum}[(a)]
          \item $k_j'=k_j$ if $v_j\not\in S$,
          \item $k_j'=k_j-1$ if~$v_j\in S\setminus S'$,
          \item $k_j'\in\{k_j-1,k_j\}$ if~$v_j\in S\cap S'$, and
          \item $l' = l - |\{j\in[n]\mid v_j\in S\cap S', k_j'=k_j\}|$.  
      \end{compactenum}
    \end{compactenum}

    Conditions~(i)--(iii) ensure that the values~$k_j$ are consistent with~$S$ and yield a local vertex cover for layer~$i$.
    Condition~(iv) checks whether the local solution extends to a valid solution for all previous layers. Note that~(c) captures the case that a solution~$\vc$ contains some~$(v_j,a,i-1)$ and~$(v_j,i,i)$ (in which case~$k_j' = k_j-1$).
    
    Thus, a table entry can be computed with at most~$4^n$ table look-ups leading to an overall running time of
    $\tau(\ell+1) (k+1)^n2^{O(n)}$.
\end{proof}

\noindent
Note that, 
\Cref{thm:XP-n-max} (i) also implies that
\MTmax parameterized by $n+k+\ell$ is in \FPT and 
that
\Cref{thm:XP-n-max} (ii) implies that
\MTsum parameterized by $n+k$ is in \FPT.
To show that \MTmax is also in \FPT when parameterized by $n+k$,
we observe that it is sufficient to consider only two possibilities to cover a time-edge.
This leads to a search tree algorithm similar to the one for the classical \textsc{Vertex Cover}.
\begin{theorem}
  \label{thm:fpt-n+k-max}
  \MTmax is solvable in~$O(2^{nk}n^2\tau)$ time.
\end{theorem}
\begin{proof}
  Let $(\G=(V,(E_i)_{i\in[\tau]}),k,\ell)$ be a \MTmax instance.
  We solve the instance with a search tree algorithm. To this end, we store counters~$k_v$, $v\in V$, for the number of intervals chosen for each vertex. We initially set all counters to~zero and start with an empty solution~$\vc=\emptyset$.

  Let~$i=\min\{t\in[\tau]\mid E_t\neq\emptyset\}$ and let $\{u,v\}\in E_i$ be an arbitrary edge (if all layers are empty, then we return ``yes'').
  If~$k_u=k$ and~$k_v=k$, then we return ``no'' since we cannot cover this edge.
  Otherwise, we simply branch into the (at most two) options of either taking~$(u,i,i')$ (if~$k_u<k$) or~$(v,i,i')$ (if~$k_v<k$) into~$\vc$, where~$i'=\min(i+\ell,\tau)$. In each branch, we increase the corresponding counter ($k_u$ or~$k_v$) by one, delete all edges incident to the corresponding vertex ($u$ or~$v$) from the layers~$E_i,\ldots,E_{i'}$, and recursively proceed on the remaining instance.
  Clearly, the recursion terminates after at most~$nk$ calls. Thus, the running time is in~$O(2^{nk}n^2\tau)$.
\end{proof}

\noindent
As a $k$-activity timeline $\vc$ for a temporal graph $\G$ with $n$ vertices is of size at most $kn$, 
\Cref{thm:fpt-n+k-max} and \Cref{thm:XP-n-max} (ii) imply that
\MTmax and \MTsum are in \FPT when parameterized by the solution size $|\vc|$.

We continue with disentangling the parameterized complexity of \MTsum and \MTmax regarding $n$ and $\ell$ in \Cref{sec:sum-n-l,sec:max-n-l}.

\section{Parameterizing \MTsum by~$n+\ell$}\label{sec:sum-n-l}

The main goal of this section is to show the following.

\begin{theorem}
	\label{thm:sum-fpt-n+l}
	\MTsum parameterized by $n+\ell$ is in \FPT.
\end{theorem}
The formal proof is deferred to the end of this section.
The general idea behind the algorithm is to split the $k$-activity timeline $\vc$
into two parts $\vc_0 \uplus \vc_{>0} = \vc$.
Here, $\vc_0$ contains all elements~$(v,a,b)\in \vc$ with~$a=b$ and
$\vc_{>0}$ contains all elements~$(v,a,b)\in \vc$ with~$a<b$.
Note that~$\sum_{(v,a,b)\in\vc_{0}}(b-a)= 0$ and 
that~$\sum_{(v,a,b)\in\vc_{>0}}(b-a)\le \ell$.
We observe that $\vc_{>0}$ induces an interval graph and 
that the number of possible induced interval graphs is upper-bounded by a function only depending on~$n+\ell$.
The algorithm iterates over all of these interval graphs and identifies the part of 
our temporal graph which can be covered by a corresponding~$\vc_{>0}$.
It remains to find a $\vc_0$ that covers the rest of our temporal graph.
To this end, we will show that \MTsum parameterized by $n$ is in \FPT{} if $\ell=0$.
In fact, as both problems are equivalent if $\ell=0$,
we show that \MTmax and \MTsum are both in \FPT when parameterized by~$n$ in this case.
In the proof of \Cref{thm:sum-fpt-n+l},
we will need a slightly more general version of these problems, 
in which there is not a single value $k$, but a value $k_v$ for each $v \in V$.
We call this generalization \textsc{Nonuniform \MTmax (\MTsum)}.

\begin{lemma}\label{thm:ell=0FPTn}
  \textsc{Nonuniform \MTmax (\MTsum)} parameterized by~$n$ is in \FPT if~$\ell=0$.
\end{lemma}

\begin{proof}
  Given an instance~$\mathcal{I}=(\G=(V,(E_i)_{i\in[\tau]}),(k_v)_{v\in V},0)$,
  we will create an integer linear program (ILP) with a number of variables bounded by some function of~$n$.
  This ILP is feasible if and only if~$\mathcal{I}$ is a yes-instance.
  By Lenstra's algorithm~\cite{L83}, this implies that \textsc{Nonuniform \MTmax (\MTsum)} parameterized by~$n$ is in \FPT, if~$\ell=0$.
  
  We use a variable $X_E^S$ for every $E\subseteq \binom{V}{2}$ and $S\subseteq V$ where $S$ is a vertex cover for~$(V,E)$.
  Note that the number of these variables is at most $2^{\binom{n}{2}+n}$.
  Intuitively, the value of the variable~$X_E^S$ gives us the number of times the vertex cover $S$ is used to cover a layer with edge set~$E$.
  
  For $E\subseteq \binom{V}{2}$, let $a(E)\coloneqq |\{ t \in [\tau]\mid E_t = E\}|$ denote the number of times the edge set~$E$ appears as a layer in~$\G$.
  Let $\vcs(E)\coloneqq \{S \subseteq V \mid S \text{ is a vertex cover of } (V,E)\}$.
  Then, the ILP constraints are as follows:
  \begin{alignat}{3}
		\sum_{S \in \vcs(E)} X_E^S & = a(E), &&\text{ for all } E\subseteq \binom{V}{2},\label{eqn:cover}\\
  		\sum_{E\subseteq \binom{V}{2}}\sum_{\substack{S \in \vcs(E) \\ \text{s.t. } v \in S}} X_E^S & \leq k_v, &&\text{ for all } v \in V,\label{eqn:k}\\
  		X_E^S &\in \mathbb{N}, &&\text{ for all } E\subseteq \binom{V}{2}, S \in \vcs(E).\nonumber
  \end{alignat}
  As we mentioned before, we are only interested in the feasibility of this ILP, so there is no objective function to optimize.

  The above ILP is feasible if and only if $\mathcal{I}$ is a yes-instance:
  Suppose that $(X_E^S)_{E\subseteq\binom{V}{2},S \in \vcs(E)}$ is a solution.
  We define a solution~$\vc$ of~$\mathcal{I}$ as follows:
  For any $E\subseteq \binom{V}{2}$, let $\vcs(E) = \{S^E_1,\ldots,S^E_{r_E}\}$.
  We use $S^E_1$ to cover the first $X_E^{S_1}$ appearances of $(V,E)$ in the layers of~$\G$, that is, we add $(v,t,t)$ to~$\vc$ for every~$v\in S^E_1$ and appearance~$(V,E_t)$.
  Then we continue with $S^E_2$ to cover the next $X_E^{S_2}$ appearances, and so on.
  Condition~(\ref{eqn:cover}) guarantees that $\vc$ covers~$\G$ and
  Condition~(\ref{eqn:k}) ensures that at most~$k_v$ intervals are chosen for each vertex~$v\in V$.
  
  Conversely, suppose that~$\mathcal{I}$ admits a solution $\vc$.
  For $t \in [\tau]$, let $X_t \coloneqq \{ v \in V \mid (v,t,t) \in \vc\}$.
  Since~$\ell=0$, $X_t$ must be a vertex cover for~$(V,E_t)$.
  Then, set $X_E^S \coloneqq |\{ t \in [\tau] \mid E_t = E \text{ and } X_t = S \}|$.
  It is easy to check that this yields a solution of the ILP.
\end{proof}

As a side result, \Cref{thm:coloring-equivalence,thm:ell=0FPTn} together imply the following.
\begin{corollary}
	\textsc{$(a{:}b)$-Coloring} parameterized by the number of vertices is in \FPT.
\end{corollary}

We are now set to show  \Cref{thm:sum-fpt-n+l}:
\MTsum parameterized by~$n+\ell$ is in \FPT.
The algorithm in the following proof is illustrated in \Cref{fig:sum-fpt-n+l}.

\begin{proof}[Proof of \Cref{thm:sum-fpt-n+l}]
	\begin{figure}[t]
		\centering
		\begin{tikzpicture}[xscale=1.2,yscale=0.8]
			\def\nsc{0.33}
			\tikzstyle{xnode}=[circle,scale=\nsc,draw,fill=black];
			\tikzstyle{singleton}=[rotate=45,rectangle,scale=3.25*\nsc,fill=gray,opacity=0.5];
			
			\node () at (0.25,1) {$v_1$};
			\node () at (0.25,2) {$v_2$};
			\node () at (0.25,3) {$v_3$};
			\node () at (0.25,4) {$v_4$};
			\node () at (0.25,5) {$v_5$};
			
			\foreach \x in {1,...,10} {
				\draw[loosely dotted] (\x+0.5,0.5) -- (\x+0.5,5.5);
			}
			
			\foreach \x in {1,...,11} {
				\foreach \y in {1,...,5} {
					\node (\x!v\y) at (\x,\y)[xnode]{};
				}
			}
			\draw (1!v1) -- (1!v2);
			\draw (1!v4) -- (1!v5);
			\draw (1!v2) .. controls (0.75,3.5) .. (1!v5);
			\draw (2!v3) -- (2!v4);
			\draw (2!v4) -- (2!v5);
			\draw (2!v2) .. controls (1.75,3) .. (2!v4);
			\draw (3!v2) -- (3!v3);
			\draw (3!v3) -- (3!v4);
			\draw (3!v4) -- (3!v5);
			\draw (3!v1) .. controls (3.25,2) .. (3!v3);
			\draw (3!v1) .. controls (2.5,2.5) .. (3!v4);
			\draw (3!v2) .. controls (2.75,3) .. (3!v4);
			\draw (4!v2) -- (4!v3);
			\draw (4!v3) -- (4!v4);
			\draw (4!v1) .. controls (3.75,2) .. (4!v3);
			\draw (4!v3) .. controls (3.75,4) .. (4!v5);
			\draw (5!v1) .. controls (5.25,2) .. (5!v3);
			\draw (5!v2) -- (5!v3);
			\draw (5!v2) .. controls (4.75,3) .. (5!v4);
			\draw (5!v3) -- (5!v4);
			\draw (5!v4) -- (5!v5);
			\draw (6!v1) -- (6!v2);
			\draw (6!v3) -- (6!v4);
			\draw (6!v4) -- (6!v5);
			\draw (6!v2) .. controls (5.75,3) .. (6!v4);
			\draw (7!v4) -- (7!v3);
			\draw (7!v4) -- (7!v5);
			\draw (7!v4) .. controls (6.75,3) .. (7!v2);
			\draw (7!v4) .. controls (7.25,2.5) .. (7!v1);
			\draw (8!v2) -- (8!v3);
			\draw (9!v2) -- (9!v3);
			\draw (9!v4) -- (9!v5);
			\draw (9!v2) .. controls (8.75,3) .. (9!v4);
			\draw (9!v2) .. controls (9.25,3.5) .. (9!v5);
			\draw (10!v1) -- (10!v2);
			\draw (10!v2) -- (10!v3);
			\draw (10!v1) .. controls (10.25,2) .. (10!v3);
			\draw (10!v1) .. controls (10.5,2.5) .. (10!v4);
			\draw (10!v1) .. controls (9.5,3) .. (10!v5);
			\draw (10!v2) .. controls (9.75,3) .. (10!v4);
			\draw (11!v1) -- (11!v2);
			\draw (11!v1) .. controls (11.25,2) .. (11!v3);
			\draw (11!v1) .. controls (11.5,2.5) .. (11!v4);
			\draw (11!v1) .. controls (10.75,3) .. (11!v5);

			\fill[rounded corners,gray,opacity=0.5] (1.8, 3.75) rectangle (3.2, 4.25) {};
			\fill[rounded corners,gray,opacity=0.5] (2.8, 2.75) rectangle (5.2, 3.25) {};
			\fill[rounded corners,gray,opacity=0.5] (4.8, 3.75) rectangle (7.2, 4.25) {};
			\fill[rounded corners,gray,opacity=0.5] (8.8, 1.75) rectangle (10.2, 2.25) {};
			\fill[rounded corners,gray,opacity=0.5] (9.8, 0.75) rectangle (11.2, 1.25) {};
			
			\fill[rounded corners,gray,opacity=0.25] (1.7,0.75) rectangle (7.3,5.25) {};
			\fill[rounded corners,gray,opacity=0.25] (8.7,0.75) rectangle (11.4,5.25) {};

			\node () at (1!v2)[singleton]{};
			\node () at (1!v5)[singleton]{};
			\node () at (6!v1)[singleton]{};
			\node () at (8!v3)[singleton]{};
			\node () at (9!v5)[singleton]{};
		\end{tikzpicture}
		\caption{
			Illustration of the algorithm for $\ell = 7$ and $k=2$ in the proof of \Cref{thm:sum-fpt-n+l}:
			There are two solution patterns with the $\ell'$-patterns highlighted in light gray and the intervals are highlighted in dark gray.
			The solution is completed with intervals of length zero, which are represented by diamonds.
		}
		\label{fig:sum-fpt-n+l}
	\end{figure}
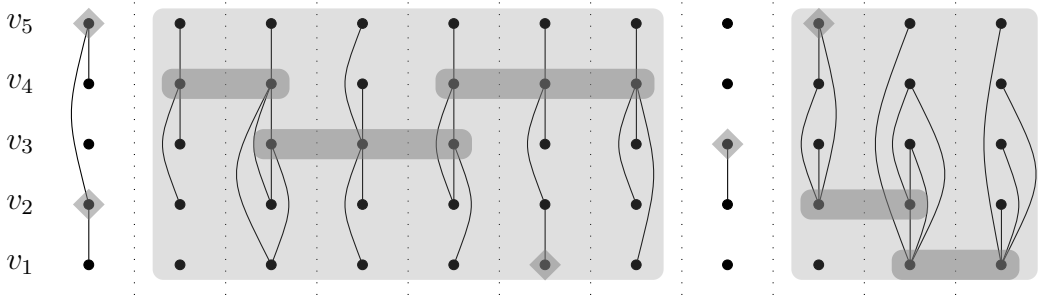
	For an instance~$(\G=(V,(E_i)_{i\in[\tau]}),k,\ell)$, we say that an \emph{$\ell'$-pattern} is a 
	sequence $(F_0,\ldots,F_{\ell'})$ of edge sets where $F_i\subseteq \binom{V}{2}$ for all $i \in [\ell']$.
	A \emph{solution pattern} is a pair $(P,\mathcal{I})$ where~$P$ is an~$\ell'$-pattern with $1 \leq \ell' \leq \ell$ and $\mathcal{I} \subseteq V \times \{0,\ldots,\ell'\} \times \{0,\ldots,\ell'\}$ with $\mathcal{I} \neq \emptyset$ and $a < b$ for all~$(v,a,b) \in \mathcal{I}$.
	Let $\mathcal{P}$ denote the set of all solution patterns.
	Note that there are at most $2^{\binom{n}{2}(\ell+1)}\cdot\ell$ different $\ell'$-patterns with $1\le \ell' \leq \ell$.
	This yields $|\mathcal{P}| \leq 2^{\binom{n}{2}(\ell+1)+n(\ell+1)^2}\cdot \ell \in 2^{O(n^2\ell^2)}$.
	The \emph{weight} of a solution pattern $(P,\mathcal{I})$ is $w((P,\mathcal{I})) \coloneqq \sum_{(v,a,b) \in \mathcal{I}} (b -a)$.
  
	Our algorithm first checks all possibilities of how often a solution ``matches'' each solution pattern~$(P,\mathcal{I})$, that is, it uses intervals as specified by~$\mathcal{I}$ to cover a subsequence of layers of~$\mathcal G$ that is equal to~$P$ (a formal explanation follows below).
	Since~$\mathcal I$ contains an interval with positive length, we can choose at most $\ell$ such solution patterns.
	Hence, we iterate over all functions~$f\colon \mathcal{P} \rightarrow \{0,\ldots,\ell\}$ with $\sum_{x \in \mathcal{P}} f(x) w(x) \leq \ell$ and which have the property that $\sum_{(P,\mathcal I)\in\mathcal P}f((P,\mathcal I))|\{ (v,a,b) \mid (v,a,b) \in \mathcal I\}|\leq k$ for all $v\in V$.
	There are at most $(\ell+1)^{|\mathcal{P}|}$ such functions.
	Next, we try out every order in which the solution patterns are matched by a solution.
	The result is a sequence $(P_1,\mathcal{I}_1),\ldots,(P_r,\mathcal{I}_r)$ of $r\le \ell$ solution patterns.
	The arguments above imply that we check at most $(\ell+1)^{2^{O(n^2\ell^2)}} \cdot \ell !$ such sequences.
	
	Next, we check whether $(P_1,\mathcal{I}_1),\ldots,(P_r,\mathcal{I}_r)$ can be matched by a solution and, if it can, compute the temporal graph that remains to be covered once we implement the solution patterns.
	The above sequence can be matched by a solution if there are $i_1 < j_1 < i_2 < j_2 < \dots < i_r < j_r \in [\tau]$ such that
	$P_s = (E_{i_s},\ldots,E_{j_s})$ for all $s \in [r]$.
	Algorithmically, we simply determine $i_s$ and $j_s$ by finding the earliest occurrence of $P_s$ in $E_1,\ldots,E_\tau$ with~$i_s \geq j_{s-1} +1$.

        Next, we compute the remaining temporal graph obtained by implementing the solution patterns $(P_1,\mathcal{I}_1),\ldots,(P_r,\mathcal{I}_r)$ on $i_1,\ldots,i_r,j_1,\ldots,j_r \in [\tau]$.
	To this end, for each $s\in[r]$ and each $t\in[i_s,j_s]$, let
	$E'_t \coloneqq
	E_t \setminus \inc\{ v\in V \mid (v,a,b) \in \mathcal{I}_s,\; i_s+a\le t \le i_s + b\}$.
	We call $\G' \coloneqq (V, E'_1,\ldots,E'_\tau)$ the \emph{residual} temporal graph.
	For each $v\in V$, let $k_v \coloneqq k - \sum_{s=1}^r|\{ (v,a,b) \mid (v,a,b) \in \mathcal{I}_s\}|$.
        
	We now build the instance~$(\G',(k_v)_{v\in V},0)$ of \textsc{Nonuniform \MTsum}.
	By \Cref{thm:ell=0FPTn}, solving this instance is \FPT when parameterized by~$n$.
        Our algorithm returns~$\texttt{true}$ if $(\G',(k_v)_{v\in V},0)$ is a yes-instance for at least one sequence $(P_1,\mathcal{I}_1),\ldots,(P_r,\mathcal{I}_r)$.
        
	It remains to show that the algorithm we described is correct.
	If the algorithm returns \texttt{true}, then this result is clearly correct.
        Assume that, for $(P_1,\mathcal{I}_1),\ldots,(P_r,\mathcal{I}_r)$, the instance~$(\G',(k_v)_{v\in V},0)$ has a solution~$\vc_0$.
        Let $\vc_{>0} \coloneqq \bigcup_{s=1}^r\{ (v,i_s+a,i_s+b) \mid (v,a,b) \in \mathcal{I}_{s}\}$.
	Then, by construction, it is easy to verify that $\vc \coloneqq \vc_0 \cup \vc_{>0}$ is a solution for $(\G,k,\ell)$.
	
	Conversely, suppose that $\vc$ is a solution for $(\G,k,\ell)$.
	Let $\vc_{>0} \coloneqq \{(v,a,b) \in \vc \mid a<b\}$.
	Consider the interval multiset $\mathcal{I} \coloneqq \{[a,b] \subseteq [\tau] \mid \exists v \in V\colon (v,a,b) \in \vc_{>0}\}$ and the corresponding interval graph $G_{\mathcal{I}}\coloneqq (\mathcal{I},E_{\mathcal{I}})$, where $E_{\mathcal{I}} \coloneqq \{\{[a,b],[a',b']\} \in \binom{\mathcal{I}}{2} \mid [a,b] \cap [a',b'] \neq \emptyset\}$.
	Since $\sum_{[a,b]\in\mathcal{I}} (b-a) \leq \ell$,
	it follows that $G_{\mathcal{I}}$ contains at most $\ell$ connected components and each of these components covers at most $\ell$ time steps.
	We will describe a solution pattern corresponding to each connected component of $G_{\mathcal{I}}$.
	For a connected component $C$ of $G_{\mathcal{I}}$, let $i \coloneqq \min_{[a,b] \in C} a$ and $j \coloneqq \max_{[a,b] \in C} b$.
	Note that $j-i + 1 \leq \ell$.
	Consider the $(j-i + 1)$-pattern $P = (E_i,\ldots,E_j)$.
	Let $\mathcal{J} \coloneqq \{(v,a-i,b-i)  \mid (v,a,b) \in \vc_{>0},\; [a,b]\in C\}$.
	Then, the solution pattern corresponding to $C$ is $(P,\mathcal{J})$.
	By listing the solution pattern for each connected component of $G_{\mathcal{I}}$ in the order in which the components appear in the interval graph, we get the solution patterns $(P_1,\mathcal{I}_1),\ldots,(P_r,\mathcal{I}_r)$ where $r$ is the number of components of $G_{\mathcal{I}}$.
	
	We claim that our algorithm returns \texttt{true} for this particular sequence of solution patterns.
	Note that our algorithm may choose different indices $i_1,\ldots,i_r,j_1,\ldots,j_r$
        for implementing $(P_1,\mathcal{I}_1),\ldots,(P_r,\mathcal{I}_r)$.
	However, the resulting residual temporal graph~$\G'$ is the same as the residual temporal graph derived from the connected components of~$G_\mathcal{I}$ up to a permutation of the layers.
	Hence, by \Cref{obs:permutation},~$\G'$ is also a yes-instance for \textsc{Nonuniform \MTsum} with $\ell = 0$.
\end{proof}

\section{Parameterizing \MTmax{} by~$n+\ell$}
\label{sec:max-n-l}

In this section, we prove that fixed-parameter tractability of \MTmax parameterized by~$n+\ell$ is (in contrast to \MTsum) unlikely.

\begin{theorem}
	\label{thm:w1-max-n}
  \MTmax parameterized by~$n$ is \Wone-hard for $\ell = 1$.
\end{theorem}

The key difference to the case $\ell = 0$ (which is in \FPT (\cref{thm:ell=0FPTn})) is that, for $\ell=1$, intervals can overlap and hence, the temporal ordering of the layers is relevant.
We prove \Cref{thm:w1-max-n} in three steps by first showing
that two generalizations of \MTmax with $\ell = 1$ are \Wone-hard when parameterized by~$n$.
The key step is to show \Wone-hardness for the following ``multicolored'' version of \MTmax.

\problemdef{Multicolored \MTmax}
{A temporal graph $\G=(V,E_1,\ldots,E_\tau)$ where $V=V_1\uplus\dots\uplus V_r$, $\ell\in\mathbb{N}_0$, and $k_1,\ldots,k_r \in \mathbb{N}$.}
{Is there an activity timeline~$\vc$ covering~$\G$ such that $\max_{(v,a,b)\in\vc}(b-a)\le \ell$ and $|\{(v,a,b) \in \vc \mid v \in V_i\}| \leq k_i$ for all $i \in [r]$?}

The proof is by reduction from \textsc{Unary Bin Packing}, where
we are given~$m$ items with sizes $s_1,\ldots,s_m\in\mathbb{N}$, a number $\beta\in \mathbb{N}$ of bins, and a bin size $B \in \mathbb{N}$ with all integers encoded in unary.
We are asked to decide if there is an assignment $f\colon [m] \rightarrow [\beta]$ of items to bins such that $\sum_{i \in f^{-1}(b)} s_i \leq B$ for all $b \in [\beta]$.
This problem is \Wone-hard when parameterized by~$\beta$~\cite{JKMS13}.
The general idea behind the reduction is that bins can be represented by vertices and the bin size by the color budgets $k_i$.

\begin{lemma}
	\label{lem:w1-mcmax-n}
	\textsc{Multicolored \MTmax} parameterized by~$n$ is \Wone-hard for $\ell = 1$.
\end{lemma}
\begin{proof}
	\begin{figure}[t]
		\centering
		\begin{tikzpicture}[xscale=0.75,yscale=0.5]
		\def\nsc{0.33}
		\tikzstyle{xnode}=[circle,scale=\nsc,draw,fill=black];
		
		\node () at (0.25,1) {$v_3$};
		\node () at (0.25,2) {$v_2$};
		\node () at (0.25,3) {$v_1$};
		\node () at (0.25,4.5) {$u_1$};
		\node () at (0.25,5.5) {$u_2$};
		\node () at (0.25,6.5) {$u_3$};
		
		\foreach \x in {1,...,17} {
			\draw[loosely dotted] (\x+0.5,0.5) -- (\x+0.5,7);
		}
		
		\foreach \x in {1,...,18} {
			\node (\x!v3) at (\x,1)[xnode]{};
			\node (\x!v2) at (\x,2)[xnode]{};
			\node (\x!v1) at (\x,3)[xnode]{};
			\node (\x!u1) at (\x,4.5)[xnode]{};
			\node (\x!u2) at (\x,5.5)[xnode]{};
			\node (\x!u3) at (\x,6.5)[xnode]{};
			
			\draw (\x!u1) -- (\x!u2);
			\draw (\x!u2) -- (\x!u3);
			\draw (\x!u1) .. controls (\x-0.25,5.5) .. (\x!u3);
		}
		\foreach \x in {2,3,6,7,8,9,14,15,16,17} {
			\draw (\x!u1) -- (\x!v1);
			\draw (\x!u2) .. controls (\x+0.35,3.75) .. (\x!v2);
			\draw (\x!u3) .. controls (\x+0.55,3.75) .. (\x!v3);
		}
		
		\draw [
		thick,
		decoration={
			brace,
			mirror,
			raise=0.4cm,
			amplitude=0.2cm
		},
		decorate
		] (0.6,1) -- (4.4,1) 
		node [pos=0.5,anchor=north,yshift=-0.55cm] {item $1$};
		\draw [
		thick,
		decoration={
			brace,
			mirror,
			raise=0.4cm,
			amplitude=0.2cm
		},
		decorate
		] (4.6,1) -- (10.4,1) 
		node [pos=0.5,anchor=north,yshift=-0.55cm] {item $2$};
		\draw [
		thick,
		decoration={
			brace,
			mirror,
			raise=0.4cm,
			amplitude=0.2cm
		},
		decorate
		] (10.6,1) -- (12.4,1) 
		node [pos=0.5,anchor=north,yshift=-0.55cm] {item $3$};
		\draw [
		thick,
		decoration={
			brace,
			mirror,
			raise=0.4cm,
			amplitude=0.2cm
		},
		decorate
		] (12.6,1) -- (18.4,1) 
		node [pos=0.5,anchor=north,yshift=-0.55cm] {item $4$};
		
		\foreach \x in {5,7,9,13,15,17} {
			\fill[rounded corners,gray,opacity=0.5] (\x-0.2, 4.25) rectangle (\x+1.2, 4.75) {};
		}
		
		\foreach \x in {1,3,11,13,15,17} {
			\fill[rounded corners,gray,opacity=0.5] (\x-0.2, 5.25) rectangle (\x+1.2, 5.75) {};
		}
		
		\foreach \x in {1,3,5,7,9,11} {
			\fill[rounded corners,gray,opacity=0.5] (\x-0.2, 6.25) rectangle (\x+1.2, 6.75) {};
		}
		\foreach \x in {2} {
			\fill[rounded corners,gray,opacity=0.5] (\x-0.2, 2.75) rectangle (\x+1.2, 3.25) {};
		}
		\foreach \x in {6,8} {
			\fill[rounded corners,gray,opacity=0.5] (\x-0.2, 1.75) rectangle (\x+1.2, 2.25) {};
		}
		\foreach \x in {14,16} {
			\fill[rounded corners,gray,opacity=0.5] (\x-0.2, 0.75) rectangle (\x+1.2, 1.25) {};
		}
		\end{tikzpicture}
		\caption{Output instance of the reduction for a \textsc{Unary Bin Packing} instance consisting of four items of sizes $2$, $3$, $1$, and $3$ with $\beta =3$ and $B=3$.
			The activity timeline indicated in the figure corresponds to an assignment that adds items $1$ and $3$ to the first bin, item $2$ to the second bin, and item $4$ to the third bin.}
		\label{fig:binpacking}
	\end{figure}
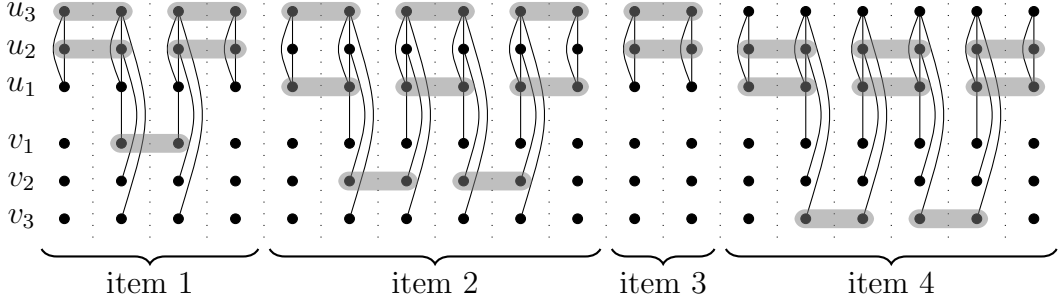
	We reduce from \textsc{Unary Bin Packing} and assume that in the input instance $(s_1,\ldots,s_m,\beta,B)$, we have $S \coloneqq \sum_{i=1}^{m} s_i = \beta B$.
	If $S > \beta B$, then this is clearly a no-instance.
	If $S < \beta B$, then we can add $\beta B - S$ items of size~$1$.

	We construct an instance $(\G=(V=V_1\uplus\dots\uplus V_r,(E_t)_{t\in [\tau]}),\ell=1,(k_i)_{i\in [r]})$ for \textsc{Multicolored \MTmax} (see \Cref{fig:binpacking}.).
	Let $V \coloneqq \{u_1,\ldots,u_\beta,v_1,\ldots,v_\beta\}$, $r \coloneqq \beta +1$, $V_i\coloneqq \{u_i\}$ for all $i\in[\beta]$, and $V_{\beta+1} \coloneqq \{v_1,\ldots,v_\beta\}$.
	We set $k_i \coloneqq S - B$ for $i\in [\beta]$ and~$k_{\beta+1} \coloneqq S - m$.
	The number of layers is $\tau \coloneqq 2 S$.
	The first $2s_1$ layers represent item $1$, followed by $2s_2$ layers representing item $2$, and so on.
	Specifically, if $t,\ldots, t+2s_i - 1$ are the layers representing item $i$, then
	\begin{align*}
			E_t &\coloneqq E_{t+2s_i - 1} \coloneqq \{ \{u_j,u_{j'}\} \mid j,j' \in [\beta], j\neq j'\} \text{ and}\\
			E_{t+a}  & \coloneqq \{ \{u_j,u_{j'}\} \mid j,j' \in [\beta], j\neq j'\}
			\cup \{\{u_j,v_j\} \mid j \in [\beta]\}
	\end{align*}
	for $a \in [2s_i -2]$.
	Clearly, this instance can be computed in polynomial time and $|V|=2\beta$.
	It remains to show that the instance for \textsc{Unary Bin Packing} is a yes-instance if and only if this output instance is a yes-instance for \textsc{Multicolored \MTmax}.
	
	Let $f$ be an assignment of items to bins such that $\sum_{i \in f^{-1}(b)} s_i = B$ for all $b \in [\beta]$.
	We will give an activity timeline $\vc$ that covers $\G$.
	For any item $i$, let $t,\ldots,t+2s_i -1$ be the~$2s_i$ layers representing this item and let $b_i \coloneqq f(i)$ be its bin.
	Then, we add the following intervals to $\vc$:
	\begin{align*}
		&\{(u_j,t+2a,t+2a+1) \mid j \in [\beta], j \neq b_i, a \in \{0,\ldots,s_i-1\}\} \cup\\
		&\{(v_{b_i},t+2a+1,t+2a+2) \mid a \in \{0,\ldots,s_i-2\}\}.
	\end{align*}
	Then, any edge $\{u_j,u_{j'}\}$ is covered because either $j\neq b_i$ or $j'\neq b_i$.
	Also, the edge $\{u_j,v_j\}$ is covered by $u_j$ for $j\neq b_i$ and by $v_{b_i}$ for $j=b_i$.
	Then, for $j \in [\beta]$, $j\neq b_i$, there are $s_i$ intervals that contain $u_j$ and no interval contains~$u_{b_i}$.
	Hence, for every $j \in [\beta]$, the vertex~$u_j$ is contained in
	\begin{align*}
		\sum_{\substack{i \in [m] \\ b_i \neq j}} s_i = S - \sum_{\substack{i \in [m] \\ b_i = j}} s_i = S - B = k_j
	\end{align*}
	intervals.
	The vertex $v_j$ is contained in
	\begin{align*}
		\sum_{\substack{i \in [m] \\ b_i = j}} (s_i -1) = B - |f^{-1}(j)|
	\end{align*}
	intervals.
	Thus, the total number of intervals containing a vertex in $V_{\beta + 1}$ is
	\begin{align*}
		\sum_{j \in [\beta]} B - |f^{-1}(j)| = \beta B - m = S -m = k_{\beta+1}.
	\end{align*}
	
	Now assume that $\vc$ is an activity timeline that covers $\G$.
	Since $\{u_1,\ldots,u_\beta\}$ induce a clique in every layer that represents an item, $\vc$ must contain all but one of these vertices in every such layer.
	There are a total of $2S$ such layers.
	Since each interval can cover a vertex in at most~$\ell+1=2$ layers, this requires $S(\beta - 1)$ intervals.
	Since only $k_j= S - B$ intervals containing a vertex $u_j$ may be chosen, the total number of intervals containing any of the vertices $u_1,\ldots,u_\beta$ is at most $\beta (S-B) = \beta S - S = S(\beta -1)$ intervals.
	Hence, in each layer that represents an item, exactly $\beta - 1$ of the vertices~$u_1,\ldots,u_\beta$ are in~$\vc$.
	Therefore, $\vc$ contains only intervals of the form $(u_b,t+a,t+a+1)$ where $t$ is the first layer corresponding to a particular item and $a$ is even.
	
	Now consider the vertices $v_1,\ldots,v_\beta$.
	Since one vertex $u_j$ is not contained in $\vc$ for every layer representing an item, the vertex $v_j$ must be in $\vc$, unless it is the first or the final layer representing that item.
	There are $\sum_{i \in [m]} (2s_i -2) = 2(S - m)$ such layers.
	Since each interval can cover at most two layers, this requires $S-m = k_{\beta + 1}$ intervals.
	Hence, none of the intervals that use any~$v_j$ can overlap.
	Therefore, $\vc$ contains intervals of the form~$(v_j,t+a,t+a+1)$ where $t$ is the first layer corresponding to a particular item and~$a$ is odd.
	
        Now, consider an item~$i\in[m]$ and the representing layers~$t,\ldots,t+2s_i-1$.
        If~$s_i=1$, then there clearly exists one~$j\in[\beta]$ such that~$(u_j,t,t+1)$ is not contained in~$\vc$.
        If~$s_i > 1$ and $\vc$ does not contain $(u_j,t+2a,t+2a+1)$ and $(u_{j'},t+2a+2,t+2a+3)$ with $j \neq j'$ and $0\le a\le s_i-2$,
	then $\vc$ must contain $(v_j,t+2a+1,t+2a+2)$ and $(v_{j'},t+2a+1,t+2a+2)$, contradicting our previous observation.
	Hence, for every item $i \in [m]$, there exists exactly one $j \in[\beta]$ such that $(u_j,t+2a,t+2a+1)$ is not in~$\vc$ for all $0 \leq a \leq s_i-1$.
	We will call this bin $b_i$.
	This yields the assignment $f(i) \coloneqq b_i$ for all $i\in [m]$.
	
	Suppose that $\sum_{i \in f^{-1}(j)} s_i > B$ for any $j \in[\beta]$.
	Since $\sum_{j\in [\beta]}\sum_{i \in f^{-1}(j)} s_i = \sum_{i\in[m]} s_i = \beta B$, this implies that $\sum_{i \in f^{-1}(j')} s_i < B$ for some $j' \in [\beta]$.
	Then, $u_{j'}$ is contained in \[\sum_{\substack{i\in[m]\\f(i)\neq j}}s_i > S - B\] intervals in~$\vc$, which is not possible.
	Hence, $f$ is an assignment of items to bins that satisfies $\sum_{i \in f^{-1}(j)} s_i = B$ for all $j\in [\beta]$.
\end{proof}

Recall the definition of \textsc{Nonuniform \MTmax} given in \Cref{sec:sum-n-l}.
Using \Cref{lem:w1-mcmax-n}, we can now show that it is \Wone-hard when parameterized by $n$ for $\ell=1$.

\begin{lemma}
	\label{lem:w1-nonuniform}
	\textsc{Nonuniform \MTmax} parameterized by~$n$ is \Wone-hard for $\ell = 1$.
\end{lemma}
\begin{proof}
	The proof is by reduction from \textsc{Multicolored \MTmax}.
	Let $(\G=(V=V_1\uplus\dots\uplus V_r,(E_t)_{t\in [\tau]}),\ell=1,(k_i)_{i\in [r]})$ be an instance of \textsc{Multicolored \MTmax}.

	We construct an instance $(\G'=(V,(E'_t)_{t\in [\tau']}),\ell=1,(k'_v)_{v\in V})$ of \textsc{Nonuniform \MTmax} as follows.
	For every $i\in[r]$ and every $v \in V_i$, we set $k'_v \coloneqq k_i$.
	Let $\tau' \coloneqq \tau + 2\sum_{i = 1}^r k_i$.
	For $t \in [\tau]$, we set $E'_t \coloneqq E_t$.
	It remains to define the additional layers $\tau+1,\ldots,\tau'$.
	For every $i \in [r]$, we add $2k_i$ layers $t_i \coloneqq \tau + 1 +2\sum_{j=1}^{i-1}k_j,\ldots,t_i + 2k_i-1$, in which the vertices in $V_i$ form a clique while all other vertices are isolated.
	
	Suppose that $\vc$ is an activity timeline for $\G$ with $|\{(v,a,b) \in \vc \mid v \in V_i\}| \leq k_i$ for all~$i \in [r]$.
	For each $v\in V$, let $\overline{k}_v \coloneqq |\{(v,a,b)\in \vc  \mid a\in[\tau]\}|$ be the number of times $v$ is used in $\vc$, that is, $\sum_{v\in V_i} \overline{k}_v \leq k_i$ for each~$i\in[r]$.
	Hence, there exists a function~$f_i \colon [k_i] \rightarrow V_i$ such that $|f_i^{-1}(v)| \geq \overline{k}_v$ for all $v\in V_i$.
        Using~$f_i$, we can cover the clique on~$V_i$ in the layers~$t_i,\ldots,t_i +2k_i-1$ that represent the color~$i$ with the intervals
        \[\vc'_i \coloneqq \{(v,t_i+2a-2,t_i + 2a-1) \mid v \in V_i, a \in [k_i], f_i(a) \neq v\}.\]
	Then, $\vc' \coloneqq \vc \cup (\bigcup_{i\in [r]} \vc'_i)$ is an activity timeline for~$\G'$.
	Moreover, for every $v\in V_i$, the following holds:
	\begin{align*}
          &|\{(v,a,b) \in \vc' \mid a \in[\tau']\}|=\\
          &|\{(v,a,b) \in \vc \mid a \in [\tau]\}| +
	|\{(v,a,a+1) \in \vc_i' \mid  \tau+1 \leq a \leq \tau' -1\}| \\
	&= \overline{k}_v + k_i - |f_i^{-1}(v)| \le k_i=k'_v.
	\end{align*}
	
	Conversely, suppose that $\vc'$ is an activity timeline for $\G'$.
	In every layer $t_i,\ldots,t_i+2k_i-1$ at least $|V_i|-1$ vertices in $V_i$ must be active since $V_i$ forms a clique in these layers.
	Since every interval can cover only two layers, this requires at least $(|V_i| -1)k_i$ intervals that use vertices in $V_i$.
	Hence, for every $i \in [r]$, $\vc'$ can contain at most $k_i$ intervals $(v,a,b)$ with~$a\in [\tau]$ and $v\in V_i$.
	Therefore, $\vc \coloneqq \{(v,a,b) \in \vc' \mid 1 \leq a \leq \tau \}$ is an activity timeline that covers $\G$ with the required property.
\end{proof}

To prove \Cref{thm:w1-max-n}, we now show how to reduce \textsc{Nonuniform \MTmax} to \MTmax.

\begin{proof}[Proof of \Cref{thm:w1-max-n}]
	We reduce from \textsc{Nonuniform \MTmax}.
	Given an input instance $(\G=(V=\{v_1,\ldots,v_n\},(E_t)_{t \in [\tau]}),\ell = 1, (k_v)_{v\in V})$,
	we construct an instance $(\G'=(V',(E'_t)_{t \in [\tau']}),\ell=1,k)$ of \MTmax.
	We let $V'\coloneqq V \cup \{u_1,u_2\}$, $\tau' \coloneqq \tau  + 2k(|V|+2)$, and $k \coloneqq \max_{v\in V} k_v$.
	The layers of $\G'$ are as follows:
	For $t \in [\tau]$, we let~$E'_t \coloneqq E_t$.
	The layers $E'_{\tau+1},\ldots,E'_{\tau+4k}$ only contain the edge $\{u_1,u_2\}$.
	Then, for~$i\in[n]$, the layers $E'_{\tau+2k(i+1) + 1},\ldots,E'_{\tau+2k(i+1) + 2(k - k_{v_i})}$ contain only the edge $\{v_i,u_1\}$, while the layers~$E'_{\tau+2k(i+1) + 2(k - k_{v_i}) +1},\ldots,E'_{\tau+2k(i+2)}$ are empty.
	
	Suppose that $\vc$ is an activity timeline that covers $\G$ and contains only~$k_v$ intervals that use $v$ for each $v\in V$.
	Then, we construct a $k$-activity timeline~$\vc'$ that covers $\G'$ as follows.
	We include all intervals in $\vc$.
	We add $(u_1,\tau+2a-1,\tau+2a)$ and $(u_2,\tau+2k+2a-1,\tau+2k+2a)$ for all $a\in [k]$ and $(v_i,\tau+2k(i+1)+a-1,\tau+2k(i+1)+a)$ for all $a \in [k-k_v]$.
	
	Now suppose that $\vc'$ is a $k$-activity timeline that covers $\G'$.
	First, $\vc'$ must contain~$k$ intervals that use $u_1$ and $k$ intervals that use $u_2$ in order to cover the appearances of the edge $\{u_1,u_2\}$ in $E'_{\tau+1},\ldots,E'_{\tau+4k}$.
	Hence, the edges $\{u_1,v_i\}$ can only be covered by intervals that use $v_i$.
	This requires~$k - k_{v_i}$ intervals that use $v_i$.
	Hence, $\vc \coloneqq \{(v,a,b) \in \vc' \mid a\in[\tau], v\in V\}$ is an activity timeline that covers $\G$ and only contains $k_v$ vertices that use each~$v\in V$.
\end{proof}

Unless the ETH fails, \textsc{Unary Bin Packing} cannot be solved in time $f(\beta)|I|^{o(\beta/\log\beta)}$ for any function~$f$, where~$|I|$ is the input size~\cite[Theorem~3]{JKMS13}.
The fact that our reduction yields a temporal graph with~$O(\beta)$ vertices implies the following:

\begin{corollary}\label{cor:ETHbin}
	\MTmax cannot be solved in time $f(n)|\G|^{o(n/\log n)}$, for any function~$f$, even if $\ell = 1$, unless the ETH fails. 
\end{corollary}

\section{Conclusion}
We completely settled the computational complexity of \MTmax regarding the considered parameters.
For \MTsum, the open question remaining is whether it is in \FPT when parameterized by the number of vertices.
Besides this question, there are many others which could be studied in future work:
\begin{compactitem}
  \item Are there (polynomial) kernelizations for the \FPT cases? Developing data reduction rules might be especially interesting from a practical perspective.
  \item Which other parameters yield tractable special cases? For example, can the number of vertices be replaced by a smaller parameter (e.g. the vertex cover number or the treewidth of the underlying graph)?
  \item What changes if one only bounds the overall number of activity intervals instead of bounding for each vertex? It should be possible to modify all our algorithms to solve this problem.
    Maybe some of our hard cases become tractable?
  \item What about ``temporalizing'' other vertex selection problems (like \textsc{Dominating Set}) in an analogous way (that is, at each time step the set of active vertices must be a valid selection for the current graph)? In fact, our dynamic programs and the ILP should easily work here as well, since they do not specifically depend on vertex covers.
\end{compactitem}

\bibliographystyle{plainnat}
\bibliography{string-short,ref}

\begin{thebibliography}{29}
\providecommand{\natexlab}[1]{#1}
\providecommand{\url}[1]{\texttt{#1}}
\expandafter\ifx\csname urlstyle\endcsname\relax
  \providecommand{\doi}[1]{doi: #1}\else
  \providecommand{\doi}{doi: \begingroup \urlstyle{rm}\Url}\fi

\bibitem[Akrida et~al.(2020)Akrida, Mertzios, Spirakis, and Zamaraev]{AMSZ20}
Eleni~C. Akrida, George~B. Mertzios, Paul~G. Spirakis, and Viktor Zamaraev.
\newblock Temporal vertex cover with a sliding time window.
\newblock \emph{J.\ Comput.\ Syst.\ Sci.}, 107:\penalty0 108--123, 2020.

\bibitem[Boehmer et~al.(2021)Boehmer, Froese, Henkel, Lasars, Niedermeier, and
  Renken]{BoehmerFHLNR21}
Niclas Boehmer, Vincent Froese, Julia Henkel, Yvonne Lasars, Rolf Niedermeier,
  and Malte Renken.
\newblock Two influence maximization games on graphs made temporal.
\newblock In \emph{Proc.\ of the 30th IJCAI}, pages 45--51. ijcai.org, 2021.

\bibitem[Bonamy et~al.(2019)Bonamy, Kowalik, Pilipczuk, Soca\l{}a, and
  Wrochna]{BKPSW19}
Marthe Bonamy, \L{}ukasz Kowalik, Micha\l{} Pilipczuk, Arkadiusz Soca\l{}a, and
  Marcin Wrochna.
\newblock Tight lower bounds for the complexity of multicoloring.
\newblock \emph{ACM Trans.\ Comput.\ Theory}, 11\penalty0 (3):\penalty0 1--19,
  2019.

\bibitem[Bumpus and Meeks(2021)]{iwocaBumpusM21}
Benjamin~Merlin Bumpus and Kitty Meeks.
\newblock Edge exploration of temporal graphs.
\newblock In \emph{Proc.\ of the 32nd IWOCA}, pages 107--121. Springer, 2021.

\bibitem[Casteigts et~al.(2021)Casteigts, Himmel, Molter, and
  Zschoche]{algorithmicaCasteigtsHMZ21}
Arnaud Casteigts, Anne{-}Sophie Himmel, Hendrik Molter, and Philipp Zschoche.
\newblock Finding temporal paths under waiting time constraints.
\newblock \emph{Algorithmica}, 83\penalty0 (9):\penalty0 2754--2802, 2021.

\bibitem[Deligkas and Potapov(2020)]{DeligkasP20}
Argyrios Deligkas and Igor Potapov.
\newblock Optimizing reachability sets in temporal graphs by delaying.
\newblock In \emph{Proc.\ of the 34th AAAI}, pages 9810--9817. {AAAI} Press,
  2020.

\bibitem[Diestel(2016)]{Die16}
Reinhard Diestel.
\newblock \emph{Graph Theory}.
\newblock Springer, 5th edition, 2016.

\bibitem[Downey and Fellows(2013)]{DF13}
Rodney~G. Downey and Michael~R. Fellows.
\newblock \emph{Fundamentals of Parameterized Complexity}.
\newblock Texts in Computer Science. Springer, 2013.

\bibitem[Erlebach et~al.(2021)Erlebach, Hoffmann, and Kammer]{jcssErlebach0K21}
Thomas Erlebach, Michael Hoffmann, and Frank Kammer.
\newblock On temporal graph exploration.
\newblock \emph{J.\ Comput.\ Syst.\ Sci.}, 119:\penalty0 1--18, 2021.

\bibitem[Fan et~al.(2021)Fan, Ju, Hou, Ye, Wan, Wang, Mei, and
  Xiong]{FanJHYWWMX21}
Yujie Fan, Mingxuan Ju, Shifu Hou, Yanfang Ye, Wenqiang Wan, Kui Wang, Yinming
  Mei, and Qi~Xiong.
\newblock Heterogeneous temporal graph transformer: An intelligent system for
  evolving android malware detection.
\newblock In \emph{Proc.\ of the 27th KDD}, pages 2831--2839. {ACM}, 2021.

\bibitem[Fellows et~al.(2018)Fellows, Jaffke, Kir{\'{a}}ly, Rosamond, and
  Weller]{FellowsJKRW18}
Michael~R. Fellows, Lars Jaffke, Aliz~Izabella Kir{\'{a}}ly, Frances~A.
  Rosamond, and Mathias Weller.
\newblock What is known about vertex cover kernelization?
\newblock In \emph{Adventures Between Lower Bounds and Higher Altitudes}, pages
  330--356. Springer, 2018.

\bibitem[Fluschnik et~al.(2019)Fluschnik, Niedermeier, Rohm, and
  Zschoche]{FNRZ19}
Till Fluschnik, Rolf Niedermeier, Valentin Rohm, and Philipp Zschoche.
\newblock Multistage vertex cover.
\newblock In \emph{Proc.\ of the 14th IPEC}, pages 14:1--14:14, 2019.

\bibitem[Heeger et~al.(2021)Heeger, Himmel, Kammer, Niedermeier, Renken, and
  Sajenko]{tcsHeegerHKNRS21}
Klaus Heeger, Anne{-}Sophie Himmel, Frank Kammer, Rolf Niedermeier, Malte
  Renken, and Andrej Sajenko.
\newblock Multistage graph problems on a global budget.
\newblock \emph{Theor.\ Comput.\ Sci.}, 868:\penalty0 46--64, 2021.

\bibitem[Holme and Saramäki(2012)]{HS12}
Petter Holme and Jari Saramäki.
\newblock Temporal networks.
\newblock \emph{Phys.\ Rep.}, 519\penalty0 (3):\penalty0 97--125, 2012.

\bibitem[Impagliazzo and Paturi(2001)]{IP01}
Russell Impagliazzo and Ramamohan Paturi.
\newblock On the complexity of {$k$-SAT}.
\newblock \emph{J.\ Comput.\ Syst.\ Sci.}, 62\penalty0 (2):\penalty0 367--375,
  2001.

\bibitem[Jansen et~al.(2013)Jansen, Kratsch, Marx, and Schlotter]{JKMS13}
Klaus Jansen, Stefan Kratsch, Dániel Marx, and Ildikó Schlotter.
\newblock Bin packing with fixed number of bins revisited.
\newblock \emph{J.\ Comput.\ Syst.\ Sci.}, 79\penalty0 (1):\penalty0 39--49,
  2013.

\bibitem[Klobas et~al.(2021)Klobas, Mertzios, Molter, Niedermeier, and
  Zschoche]{ijcaiKlobasMMNZ21}
Nina Klobas, George~B. Mertzios, Hendrik Molter, Rolf Niedermeier, and Philipp
  Zschoche.
\newblock Interference-free walks in time: Temporally disjoint paths.
\newblock In \emph{Proc.\ of the 30th IJCAI}, pages 4090--4096. ijcai.org,
  2021.

\bibitem[Lenstra(1983)]{L83}
Hendrik~W. Lenstra.
\newblock Integer programming with a fixed number of variables.
\newblock \emph{Math.\ Oper.\ Res.}, 8\penalty0 (4):\penalty0 538--548, 1983.

\bibitem[Lewis and Yannakakis(1980)]{LY80}
John~M. Lewis and Mihalis Yannakakis.
\newblock The node-deletion problem for hereditary properties is {NP}-complete.
\newblock \emph{J.\ Comput.\ Syst.\ Sci.}, 20\penalty0 (2):\penalty0 219--230,
  1980.

\bibitem[Mertzios et~al.(2020)Mertzios, Molter, Niedermeier, Zamaraev, and
  Zschoche]{stacsMertziosMNZZ20}
George~B. Mertzios, Hendrik Molter, Rolf Niedermeier, Viktor Zamaraev, and
  Philipp Zschoche.
\newblock Computing maximum matchings in temporal graphs.
\newblock In \emph{Proc.\ of the 37th STACS}, pages 27:1--27:14, 2020.

\bibitem[Mertzios et~al.(2021)Mertzios, Molter, and Zamaraev]{jcssMertziosMZ21}
George~B. Mertzios, Hendrik Molter, and Viktor Zamaraev.
\newblock Sliding window temporal graph coloring.
\newblock \emph{J.\ Comput.\ Syst.\ Sci.}, 120:\penalty0 97--115, 2021.

\bibitem[Michail and Spirakis(2016)]{tcsMichailS16}
Othon Michail and Paul~G. Spirakis.
\newblock Traveling salesman problems in temporal graphs.
\newblock \emph{Theor.\ Comput.\ Sci.}, 634:\penalty0 1--23, 2016.

\bibitem[Nederlof(2008)]{Nederlof08}
Jesper Nederlof.
\newblock Inclusion exclusion for hard problems, 2008.
\newblock Master thesis. Department of Information and Computer Science,
  Utrecht University.

\bibitem[Razgon and O'Sullivan(2009)]{RO09}
Igor Razgon and Barry O'Sullivan.
\newblock Almost {2-SAT} is fixed-parameter tractable.
\newblock \emph{J.\ Comput.\ Syst.\ Sci.}, 75\penalty0 (8):\penalty0 435--450,
  2009.

\bibitem[Rozenshtein and Gionis(2019)]{RG19}
Polina Rozenshtein and Aristides Gionis.
\newblock Mining temporal networks.
\newblock In \emph{Proc.\ of the 25th KDD}, pages 3225--3226. ACM, 2019.

\bibitem[Rozenshtein et~al.(2017)Rozenshtein, Tatti, and Gionis]{RTG17}
Polina Rozenshtein, Nikolaj Tatti, and Aristides Gionis.
\newblock The network-untangling problem: From interactions to activity
  timelines.
\newblock In \emph{Proc.\ of the ECML/PKDD~'17}, pages 701--716. Springer,
  2017.

\bibitem[Rozenshtein et~al.(2021)Rozenshtein, Tatti, and Gionis]{RTG21}
Polina Rozenshtein, Nikolaj Tatti, and Aristides Gionis.
\newblock The network-untangling problem: from interactions to activity
  timelines.
\newblock \emph{Data Min.\ Knowl.\ Discov.}, 35\penalty0 (1):\penalty0
  213--247, 2021.

\bibitem[Singer et~al.(2019)Singer, Guy, and Radinsky]{SingerGR19}
Uriel Singer, Ido Guy, and Kira Radinsky.
\newblock Node embedding over temporal graphs.
\newblock In \emph{Proc.\ of the 28th IJCAI}, pages 4605--4612. ijcai.org,
  2019.

\bibitem[Zschoche et~al.(2020)Zschoche, Fluschnik, Molter, and
  Niedermeier]{jcssZschocheFMN20}
Philipp Zschoche, Till Fluschnik, Hendrik Molter, and Rolf Niedermeier.
\newblock The complexity of finding small separators in temporal graphs.
\newblock \emph{J.\ Comput.\ Syst.\ Sci.}, 107:\penalty0 72--92, 2020.

\end{thebibliography}

\end{document}